\documentclass[runningheads]{llncs}

\usepackage{amsmath,amssymb}
\usepackage{mathtools}
\usepackage{microtype}
\usepackage{algorithm}
\usepackage{algorithmic}
\usepackage{xspace}

\newcommand{\R}{\mathbb{R}}
\newcommand{\B}{B}
\newcommand{\Lam}{\Lambda}
\newcommand{\bars}{\B(M)}

\newcommand{\TopK}{\mathrm{Top}\text{-}K}

\newcommand{\Kc}{K_{\text{all}}}

\newcommand{\ThmA}{Theorem~\ref{thm:anchored-grouping}}
\newcommand{\ThmB}{Theorem~\ref{thm:topk-correctness-complexity}}
\newcommand{\ThmE}{Theorem~\ref{thm:stability-topk}}
\newcommand{\PropF}{Proposition~\ref{prop:lb}}

\begin{document}

\title{Top-K Exterior Power Persistent Homology: Algorithm, Structure, and Stability}
\titlerunning{Top-K Exterior Power Persistent Homology}
\author{Yoshihiro Maruyama}
\authorrunning{Y. Maruyama}
\institute{School of Informatics, Nagoya University, Japan\\
\email{maruyama@i.nagoya-u.ac.jp}
}

\maketitle

\begin{abstract}
Exterior powers play important roles in persistent homology in computational geometry. 
In the present paper we study the problem of extracting the $K$ longest intervals of the exterior-power layers $\Lam^i M$ of a tame persistence module $M$, directly from the barcode $\B(M)$, without enumerating the entire $\B(\Lam^i M)$. 
We prove a structural decomposition theorem that organizes $\B(\Lam^i M)$ into monotone per-anchor streams with explicit multiplicities, enabling a best-first algorithm.
We provide an $O\bigl((M+K)\log M\bigr)$ time algorithm for any fixed $i \ge 2$, obtained via a grouped best-first search. 
We also show that the Top-$K$ length vector is $2$-Lipschitz under bottleneck
perturbations of the input barcode, and prove a comparison-model
lower bound implying the $O(M\log M)$ preprocessing is information-theoretically unavoidable. 
Our experiments confirm the theory, showing speedups over full enumeration in high overlap cases.
By enabling efficient extraction of the most prominent features, our approach makes higher-order persistence feasible for large datasets and thus broadly applicable to machine learning, data science, and scientific computing.

\end{abstract}

\section{Introduction}

Exterior powers $\Lam^i$ in persistent homology capture higher-order interactions among topological features that ordinary persistence cannot record \cite{Ghrist08,EdelsbrunnerHarer10,Oudot15}.\footnote{$\Lam$ denotes exterior product/power. $\Lam$ is defined pointwise for persistence modules (i.e., take exterior product at each element of the domain interval; more detail below).} While standard persistence tracks the lifetimes of individual cycles, exterior powers $\Lam^i$ encode how groups of cycles coexist, providing richer invariants that are stable and computable \cite{CohenSteiner07,ChazalOudot14,ChazalStructure16}. These higher-order signatures are increasingly important in applications ranging from theoretical topology to machine learning, where concise and discriminative summaries are essential.

However, computing the full exterior-power barcode quickly becomes infeasible: even for $\Lam^2$, its size can be quadratic in the input. Straightforward algorithms that enumerate all intervals therefore waste work when only the most significant features are needed. This motivates the \emph{Top-$K$ problem}: extracting just the $K$ longest intervals of $\B(\Lam^i M)$ without enumerating the entire structure. Such a view aligns with practice, where users seek concise visual summaries, robust statistics, or fixed-length features for downstream learning.

This work shows that Top-$K$ for exterior powers admits an efficient, stable solution, bridging classical persistent homology algorithms \cite{ELZ02,ZC05,BauerRipser21} with techniques from selection algorithms \cite{Fagin03,FredericksonJohnson84} and persistent data structures \cite{DriscollEtAl89}. This makes higher-order persistence more scalable and broadly applicable across computational geometry and topological data analysis.

The main contributions of this paper are as follows: (i) We give a structural decomposition of $\B(\Lam^i M)$ into simple monotone streams, making explicit where higher-order intervals originate and how their multiplicities arise; (ii) We design a best-first algorithm that extracts the exact Top-$K$ intervals in near output-sensitive time, avoiding the cost of full enumeration; (iii) We show that the Top-$K$ length vector is stable under bottleneck perturbations, providing a concise and noise-robust summary; (iv) We establish an $\Omega(M\log M)$ lower bound in the comparison model, proving that our preprocessing cost is optimal up to constants.

In the rest of the paper, we develop mathematical foundations first, and provide structure theorems and algorithms with complexity results. We then prove stability and optimality results. Finally, we give an experimental verification.

\section{Basic Concepts and Fundamental Theorems}\label{sec:prelim}
%We fix notation, the computational model, and the sweep/tie conventions. We also give the exterior-power interval calculus and stability property used later. 
Let $i\ge 2$ be a fixed constant throughout the paper.
%\subsection{Input model, bars, and tie breaking}

Let $M$ be a tame, pointwise finite-dimensional persistence module over
$I \subset \R$ with barcode $\B(M) \;=\; \{\, J_r = [b_r,d_r) \,\}_{r=1}^M$,  where, for convenience of notation, we use the same symbol $M$ for the module and the corresponding number of intervals (for persistence module basics, see \cite{Oudot15,ChazalStructure16}).
We write $M_t$ as usual for the value of $M$ at $t\in I$. 
$\Lambda^i M$ is defined pointwise by $(\Lambda^iM)_t=\Lambda^i M_t$.\footnote{
Note that $\Lambda^i M_t$ on the right-hand side denotes the $i$-th exterior power of the vector space $M_t$.
On morphisms, for each $s\le t$ in $I$, we set
$(\Lambda^i M)_{s\le t} \;=\; \Lambda^i(M_{s\le t}):\ 
\Lambda^i(M_s) \longrightarrow \Lambda^i(M_t).$
For persistence homology basics, see also \cite{EdelsbrunnerHarer10,Oudot15}.
}

Unless otherwise stated, we assume all bars are finite ($d_r<\infty$); this
covers common filtrations on finite complexes/graphs. 
%\paragraph{Infinite bars.}
In some filtrations (e.g.\ $H_0$ of Vietoris--Rips or \v{C}ech complexes),
some bars extend to $+\infty$. Our algorithm and decomposition extend
unchanged if such bars are handled in either of the following standard ways:
(i) \emph{Truncation:} fix a global time horizon $t_{\max}$ and replace each
  infinite bar $[b_r,\infty)$ by $[b_r,t_{\max})$, so that $\Lam^i$-intervals
  respect the finite horizon;
(ii) \emph{Relative formulation:} regard an infinite bar as persisting until
  a formal symbol $\infty$, and observe that in the exterior-power interval
  calculus (Theorem~\ref{thm:interval-calculus} below) only $\min\{d_r,\ldots\}$ appears, so truncating to any
  sufficiently large finite cutoff yields the same Top-$K$ results.
Thus we may assume without loss of generality that all bars are finite.

We adopt the \emph{closed–open} convention and process a global event list of
all $b_r$ and $d_r$ \emph{sorted by time}, breaking ties by handling
\emph{deaths before births} \cite{EdelsbrunnerHarer10,ChazalStructure16}.\footnote{In persistent homology, the \emph{birth} of a bar is the parameter 
value at which a homology class first appears, and its \emph{death} is the 
value at which that class disappears 
\cite{EdelsbrunnerHarer10,ChazalStructure16}.} 
Among births at the same time we fix any total
order that is consistent across the sweep (cf.\
\cite{EdelsbrunnerHarer10,ELZ02}).
For an event sweep from $-\infty$ to $+\infty$, just \emph{before} the birth of
bar $r$ at time $b_r$, define the \emph{alive set}
$A_r \;:=\; \{\, s : [b_s,d_s) \ \text{is alive just before } b_r \,\}$ and 
$c_r := |A_r|.$
Order $A_r$ by non-increasing death times; write these as
$d_r(1) \ge d_r(2) \ge \cdots \ge d_r(c_r)$.
For $j \in \{1,\dots,c_r\}$ set
\begin{equation}\label{eq:anchored-length}
\ell_r(j)\ :=\ \max\Bigl\{\,0,\ \min\{d_r,\ d_r(j)\} - b_r \Bigr\},
\end{equation}
so $j \mapsto \ell_r(j)$ is non-increasing.

\subsection{Exterior powers at the barcode level}
We first prove a fundamental theorem clarifying the barcode structure of exterior powers. 
\begin{theorem}[Exterior-power interval calculus]\label{thm:interval-calculus}
For any $i\ge 1$, the barcode of $\Lam^i M$ is the multiset
\[
\B(\Lam^i M)\;=\;\Bigl\{\,\bigl[\max_j b_{\ell_j},\, \min_j d_{\ell_j}\bigr)\ :\ 
\ell_1<\cdots<\ell_i,\ \max_j b_{\ell_j}<\min_j d_{\ell_j}\Bigr\}.
\]
\end{theorem}

\begin{proof}
By the barcode decomposition for tame pointwise finite-dimensional modules
(see \cite{Oudot15,ChazalStructure16}), there is a (noncanonical) isomorphism
$M \;\cong\; \bigoplus_{r=1}^M I_{[b_r,d_r)},$ 
where $I_{[b_r,d_r)}$ is the (one-dimensional) interval module supported on
$[b_r,d_r)$.

Fix $i\ge 1$. Apply the exterior-power functor pointwise in $t\in\R$.
For any finite family of vector spaces, $\Lam^i\!\Bigl(\bigoplus_{r=1}^M V_r\Bigr)
\;\cong\;
\bigoplus_{\substack{\ell_1<\cdots<\ell_i\\ \ell_j\in\{1,\dots,M\}}}
\; V_{\ell_1}\wedge\cdots\wedge V_{\ell_i}$
is natural in the $V_r$. Evaluating at time $t$ with $V_r=(I_{[b_r,d_r)})_t$ yields
\[
\Lam^i(M_t)
\;\cong\;
\bigoplus_{\ell_1<\cdots<\ell_i}
\; (I_{[b_{\ell_1},d_{\ell_1})})_t \wedge \cdots \wedge (I_{[b_{\ell_i},d_{\ell_i})})_t .
\]
Since each $(I_{[b_\ell,d_\ell)})_t$ is either $0$ (if $t\notin[b_\ell,d_\ell)$) or a
$1$-dimensional $k$ (if $t\in[b_\ell,d_\ell)$), the summand for
$\{\ell_1,\dots,\ell_i\}$ is $k$ precisely when $t$ lies in the intersection
$\bigcap_{j=1}^i [b_{\ell_j},d_{\ell_j})$, and is $0$ otherwise. Hence, as $t$
varies, the subfunctor generated by this summand is the interval module supported on
$\bigcap_{j=1}^i [b_{\ell_j},d_{\ell_j})
\;=\;
\bigl[\max_j b_{\ell_j},\,\min_j d_{\ell_j}\bigr),$ 
which is nonzero exactly when $\max_j b_{\ell_j}<\min_j d_{\ell_j}$.

Naturality of the above isomorphisms with respect to the structure maps of
$M$ shows that these pointwise decompositions assemble to an isomorphism of
persistence modules
$\Lam^i M
\;\cong\;
\bigoplus_{\ell_1<\cdots<\ell_i}
I_{\,[\max_j b_{\ell_j},\,\min_j d_{\ell_j})},$ 
with the convention that empty intersections contribute the zero module and
hence no bar. Therefore the barcode of $\Lam^i M$ is precisely the stated
multiset of intervals.
\end{proof}

\subsection{Top-$K$ with multiplicity}\label{sec:topk-mult}
Let the multiset of lengths of $\B(\Lam^i M)$ (counted with multiplicity) 
be sorted in non-increasing order as $L_1\ge L_2\ge\cdots$. For $K\ge 1$, the
\emph{Top-$K$ multiset} is $\{L_1,\dots,L_K\}$ (with multiplicity), and the
\emph{Top-$K$ length vector} is
\[
\mathbf L_K(M,i)\ :=\ (L_1,\dots,L_K)\in\mathbb{R}_{\ge 0}^K,
\]
padded with zeros if necessary.\footnote{When multiple intervals have equal length, any ordering of the ties is acceptable, since our results concern the Top-$K$ multiset and the sorted length vector $\mathbf L_K(M,i)$, which are invariant under tie-breaking (cf. top-$K$ aggregation \cite{Fagin03}).}

%\footnote{As to tie policy,  when two lengths are equal, we break ties arbitrarily and do not require a stable ordering across anchors: any permutation of equal values is acceptable. In particular, Algorithm~\ref{alg:topk} below may emit all $w_i(j)$ copies of a tied length coming from one anchor before emitting the same length from another anchor; under this policy, all results about the Top-$K$ \emph{multiset} and the vector $\mathbf L_K(M,i)$ remain invariant (cf.\ top-$K$ aggregation \cite{Fagin03}).}

%\subsection{Multiplicity shorthands and per-anchor grouping}

For $j\ge i-1$ define the binomial weight
$w_i(j)\ :=\ \binom{j-1}{\,i-2\,}.$
In Section~\ref{sec:structure} we prove that, for fixed anchor $r$, all $\Lam^i$ intervals whose \emph{largest chosen rank} equals $j$ have common length $\ell_r(j)$ and total multiplicity $w_i(j)$; moreover, the union over all anchors gives the full multiset $\B(\Lam^i M)$.

\subsection{Bottleneck distance and stability}
Let $\Delta=\{(t,t)\in\R^2:t\in\R\}$ denote the diagonal. For barcodes $X,Y$ (finite multisets of points $(b,d)$ with $b<d$), an \emph{$\varepsilon$-matching} is a partial matching between $X\cup\Delta$ and $Y\cup\Delta$ such that matched pairs are within $L_\infty$-distance $\le\varepsilon$ and unmatched points lie within $\varepsilon$ of the diagonal. The \emph{bottleneck distance} is
\[
d_B(X,Y)\ :=\ \inf\{\ \varepsilon\ge 0\ :\ \text{there exists an $\varepsilon$-matching}\ \}.
\]

\begin{theorem}[Stability of exterior powers]\label{thm:stability-exterior}
For $i\ge 1$ and tame persistence modules $M,M'$, 
\[
   d_B\!\bigl(\B(\Lam^i M),\, \B(\Lam^i M')\bigr)
   \;\le\; d_B\!\bigl(\B(M),\, \B(M')\bigr).
\]
\end{theorem}

\begin{proof}
Suppose $M$ and $M'$ are $\varepsilon$-interleaved, i.e.\ there exist
linear maps $f_t: M_t \to M'_{t+\varepsilon}$ and
$g_t: M'_t \to M_{t+\varepsilon}$ commuting with structure maps and
satisfying the usual zig–zag relations up to shift~$2\varepsilon$.
Because $\Lam^i$ is a functor on vector spaces that preserves
linear maps, applying $\Lam^i$ to each $f_t,g_t$ yields natural
transformations
\[
   \Lam^i f_t: \Lam^i M_t \;\to\; \Lam^i M'_{t+\varepsilon},
   \qquad
   \Lam^i g_t: \Lam^i M'_t \;\to\; \Lam^i M_{t+\varepsilon}.
\]
These commute with the induced structure maps of $\Lam^i M$ and
$\Lam^i M'$, since functors preserve commutative diagrams.
Moreover, the zig–zag identities are preserved under $\Lam^i$,
because if $g_{t+\varepsilon}\circ f_t$ equals the shift map
$M_t \to M_{t+2\varepsilon}$, then
$\Lam^i g_{t+\varepsilon}\circ \Lam^i f_t$
equals the shifted map on $\Lam^i M_t$.
Thus $\Lam^i M$ and $\Lam^i M'$ are also $\varepsilon$-interleaved.
By the fundamental isometry theorem of persistence
(\cite{ChazalOudot14,ChazalStructure16}), interleaving distance equals
bottleneck distance on barcodes. Hence
$d_B\!\bigl(\B(\Lam^i M),\B(\Lam^i M')\bigr)
   \;\le\; d_B\!\bigl(\B(M),\B(M')\bigr).$
\end{proof}

This result shows that exterior powers preserve the classical stability of persistence, ensuring that higher-order interaction features remain robust under perturbations of the input data.

\subsection{Computational model and data structures}
We analyze running time in the RAM (Random Access Machine) model with comparisons; sorting $O(M)$ endpoints costs $O(M\log M)$, which is information-theoretically unavoidable in this model (cf.\ Section~\ref{sec:lowerbound}; see also \cite{KnuthVol3,CLRS09}). We use \emph{coordinate compression} of distinct death times to $\{1,\dots,N\}$.

\paragraph{Persistent order-statistics (OS) tree.}
We will use a standard persistent segment tree over $\{1,\dots,N\}$ storing
\emph{counts of alive bars} at each compressed death coordinate (cf.\
\cite{DriscollEtAl89}). It supports:
\begin{itemize}
  \item \textsc{Update}$(\text{root},\ \text{pos},\ \pm 1)$ in $O(\log M)$ time, returning a new root and keeping the old root immutable;
  \item \textsc{KthFromRight}$(\text{root},\ k)$: returns the death value of the $k$-th alive
        \emph{bar} in non-increasing order, counting with multiplicity. Equivalently,
        the tree maintains cumulative counts of alive bars at each coordinate,
        and the query walks these counts from the right to locate the $k$-th bar; 
        %This matches the definition $d_r(j)$ above;
  \item \textsc{Size}$(\text{root})$ in $O(1)$ time.
\end{itemize}
During the sweep we store, for each birth of $r$, the snapshot root $T_r$
\emph{before} inserting $r$ (encoding $A_r$) and the integer
$c_r=\textsc{Size}(T_r)$. The total space is $O(M)$ for heap buffers plus
$O(M\log M)$ nodes for persistence.

\section{Algorithmic Structural Decomposition Theorem}\label{sec:structure}

We now derive a birth-anchored, rank–grouped description of $\B(\Lam^i M)$ that will drive our best‑first algorithm. Throughout this section, the sweep/tie conventions of Section~\ref{sec:prelim} apply (see also \cite{Ghrist08,Oudot15} for background on barcode manipulations).

\subsection{Anchors and alive sets}
Given an $i$-tuple $I=\{\ell_1,\dots,\ell_i\}\subseteq\{1,\dots,M\}$ with
\(
\max_j b_{\ell_j}<\min_j d_{\ell_j},
\)
let
\[
t^\star \ :=\ \max_j b_{\ell_j}\,,\qquad
A^\star\ :=\ \{\,\ell_j:\ b_{\ell_j}=t^\star\,\}.
\]
By our event order (deaths before births, and a fixed total order among equal-time births), there is a \emph{unique} index
\[
r^\star\in A^\star\quad\text{that is processed last at time }t^\star.
\]
We call $r^\star$ the \emph{anchor} of $I$. At the moment just \emph{before} $b_{r^\star}=t^\star$, all elements of $I\setminus\{r^\star\}$ are alive; hence $I\setminus\{r^\star\}\subseteq A_{r^\star}$, where $A_{r^\star}$ is the alive set from Section~\ref{sec:prelim}. Ordering $A_{r^\star}$ by 
non-increasing death times, write the deaths as
\(
d_{r^\star}(1)\ge d_{r^\star}(2)\ge\cdots\ge d_{r^\star}(c_{r^\star}).
\)

\begin{lemma}\label{lem:anchor-unique}
Assume the sweep order processes deaths before births at equal times and fixes a total order among simultaneous births.  
For any $\Lam^i$ interval arising from $I=\{\ell_1,\dots,\ell_i\}$ as in Theorem~\ref{thm:interval-calculus}, let $t^\star=\max_j b_{\ell_j}$.  
Then there is a unique anchor $r^\star$, namely the index processed last among the births at $t^\star$, and $I\setminus\{r^\star\}\subseteq A_{r^\star}$.  
Conversely, for any $r$ and any $(i{-}1)$-subset of $A_r$, the $i$-tuple $\{r\}\cup S$ yields a (possibly zero-length) $\Lam^i$ interval with birth $b_r$.
\end{lemma}

\begin{proof}
By Theorem~\ref{thm:interval-calculus}, the $\Lambda^i$ interval from 
$I=\{\ell_1,\dots,\ell_i\}$ is 
$[\max_j b_{\ell_j},\,\min_j d_{\ell_j})$. 
Let $t^\star=\max_j b_{\ell_j}$. Among the indices with birth $t^\star$, exactly 
one is processed last under the tie rule; call it $r^\star$. 
At time $b_{r^\star}$, all other $\ell_j$ are alive, so 
$I\setminus\{r^\star\}\subseteq A_{r^\star}$. 
Conversely, for any anchor $r$ and any $(i-1)$-subset $S\subseteq A_r$, 
the $i$-tuple $\{r\}\cup S$ yields the interval 
$[b_r,\min\{d_r,d_s:s\in S\})$ by the same formula.
\end{proof}

\subsection{Rank–grouping at a fixed anchor}
Fix an anchor $r$ with alive set $A_r$ of size $c_r$. Let
$S=\{j_1<\cdots<j_{i-1}\}\subseteq \{1,\dots,c_r\}$ 
be the ranks of the chosen neighbors (so the corresponding death times are $d_r(j_1),\dots,d_r(j_{i-1})$). The $\Lam^i$ interval produced by $(r,S)$ is
\[
I(r,S)\ =\ \bigl[b_r,\ \min\{d_r,\ d_r(j_1),\dots,d_r(j_{i-1})\}\bigr),
\]
whose length equals, by definition \eqref{eq:anchored-length},
$|I(r,S)|\ =\ \ell_r(\max S).$
Hence the length depends only on the \emph{largest} rank in $S$.

\begin{proposition}\label{prop:multiplicity-fixed-largest}
Fix $r$ and a rank $j\in\{i-1,\dots,c_r\}$. The number of $(i-1)$-subsets $S$
of ranks with $\max S=j$ equals $w_i(j)=\binom{j-1}{i-2}$. All such subsets yield
the same length $\ell_r(j)$, truncated below by $0$ as in~\eqref{eq:anchored-length}.
In particular, when $d_r(j)\le b_r$ the resulting value is $\ell_r(j)=0$, which
contributes nothing to $\B(\Lam^i M)$.
\end{proposition}

\begin{proof}
To have $\max S=j$, one must include rank $j$ and choose the remaining $i{-}2$ ranks from $\{1,\dots,j-1\}$, giving $\binom{j-1}{i-2}$ choices. The shared length follows because $|I(r,S)|=\min\{d_r,d_r(j)\}-b_r$ depends only on the largest rank (and is truncated below by $0$ as in \eqref{eq:anchored-length}).
\end{proof}

\begin{theorem}[Anchored rank--grouping]\label{thm:anchored-grouping}
Assume bars are closed--open and ties are broken by processing deaths before births,
with a fixed total order among simultaneous births. For each birth $r$, let $A_r$
be the set of bars alive just before $b_r$, ordered by non--increasing death
time $d_r(1)\ge\cdots\ge d_r(c_r)$. Define
\[
  \ell_r(j):=\max\{0,\ \min\{d_r,d_r(j)\}-b_r\},\qquad
  J_r:=\{\,j\in\{i-1,\dots,c_r\}: \ell_r(j)>0\,\}.
\]
Then every interval of $\B(\Lambda^i M)$ arises uniquely from a pair $(r,S)$,
where $r$ is the anchor (the last-processed birth at
$t^\star=\max_j b_{\ell_j}$) and $S\subseteq A_r$ has $|S|=i-1$,
with length $|I(r,S)|=\ell_r(\max S)$. Consequently, for each anchor $r$ the
anchored multiset equals
\[
  \{\, \ell_r(j)\ \text{with multiplicity } w_i(j)=\tbinom{j-1}{i-2}
    : j\in J_r \,\},
\]
with $j\mapsto \ell_r(j)$ non--increasing, and globally
\[
  \B(\Lambda^i M)\;=\;\bigsqcup_{r=1}^M
  \{\, \ell_r(j)\ \text{with multiplicity } w_i(j): j\in J_r \,\},
\]
a disjoint union of multisets of finite (positive-length) intervals.
\end{theorem}

\begin{proof}
Fix $i\ge 2$. By the interval calculus (Theorem~\ref{thm:interval-calculus}),
any $i$-tuple $I=\{\ell_1,\dots,\ell_i\}$ with
$\max_j b_{\ell_j}<\min_j d_{\ell_j}$ produces the $\Lam^i$–interval
$[\max_j b_{\ell_j},\ \min_j d_{\ell_j})$.
Let $t^\star=\max_j b_{\ell_j}$ and choose the unique index $r^\star$ that is
processed last among those with $b_{\ell_j}=t^\star$. Then
$I\setminus\{r^\star\}\subseteq A_{r^\star}$, and $r^\star$ is the
\emph{anchor} of $I$. Conversely, for any anchor $r$ and
$(i{-}1)$-subset $S\subseteq A_r$, the $i$-tuple $\{r\}\cup S$ yields
\[
I(r,S)=\bigl[\, b_r,\ \min\{d_r,\ d_s: s\in S\}\,\bigr).
\]
Thus every element of $\B(\Lam^i M)$ arises uniquely from some pair $(r,S)$.

Now order $A_r$ by non–increasing death times
$d_r(1)\ge d_r(2)\ge\cdots\ge d_r(c_r)$.
If $S=\{j_1<\cdots<j_{i-1}\}$ are the ranks of the chosen neighbors, then
the length of $I(r,S)$ depends only on the largest rank:
$|I(r,S)| \;=\;\ell_r(\max S).$
Therefore all $(i{-}1)$-subsets with the same maximal rank $j$ yield
the same length $\ell_r(j)$.

At this point Proposition~\ref{prop:multiplicity-fixed-largest} applies:
it tells us that the number of such subsets is exactly
$w_i(j)=\binom{j-1}{i-2}$, and that they all contribute the same
value $\ell_r(j)$ (truncated at $0$).
Hence for each anchor $r$, the multiset of anchored lengths is
\[
\bigl\{\, \ell_r(j)\ \text{with multiplicity } w_i(j) \;:\;
  j\in\{i-1,\dots,c_r\}\,\bigr\}.
\]

Finally, define $J_r=\{j: \ell_r(j)>0\}$. Restricting to $j\in J_r$
removes the zero-length intervals, which do not belong to
$\B(\Lam^i M)$. Because each interval has a unique anchor,
the global barcode is the disjoint multiset union over anchors:
\[
\B(\Lam^i M)\;=\;\bigsqcup_{r=1}^M
   \bigl\{\, \ell_r(j)\ \text{with multiplicity } w_i(j)\ :\ j\in J_r \,\bigr\}.
\]
Monotonicity of $j\mapsto \ell_r(j)$ follows directly from
the ordering $d_r(1)\ge d_r(2)\ge\cdots$.
This proves the theorem.
\end{proof}

When $i=2$, $w_2(j)=\binom{j-1}{0}=1$, so each rank contributes exactly one element and the anchored stream becomes a simple non-increasing sequence $\ell_r(1)\ge \ell_r(2)\ge\cdots\ge \ell_r(c_r)$.

The above theorem gives a complete and nonredundant
decomposition of $\B(\Lambda^i M)$ into per-anchor monotone streams with
closed-form multiplicities, providing the structural foundation for efficient
Top-$K$ algorithms and showing exactly how higher-order intervals are organized.

\section{The \textsc{TopK}--$\Lam^i$ Algorithm}\label{sec:algorithm}

We now give a best‑first algorithm that outputs the $K$ longest elements of $\B(\Lam^i M)$ \emph{without} enumerating the entire multiset.

\subsection{Preprocessing: sweep and persistent order statistics}
Build the global event list of all births and deaths, sorted by time, with \emph{deaths before births} at ties (Section~\ref{sec:prelim}; cf.\ \cite{ELZ02,EdelsbrunnerHarer10}). Coordinate-compress distinct death times to $\{1,\dots,N\}$ and maintain a \emph{persistent} order‑statistics tree over this axis, storing counts of alive deaths.
During the sweep:
\begin{itemize}
  \item On a death of bar $x$, perform an update $-1$ at the index of $d_x$.
  \item On a birth of bar $r$ at time $b_r$, \emph{before} inserting $r$:
    store the current snapshot root $T_r$ encoding $A_r$, and record $c_r=\textsc{Size}(T_r)$; then insert $+1$ at the index of $d_r$ so that $r$ is alive for later anchors.
\end{itemize}
This costs $O(M\log M)$ time and $O(M\log M)$ persistent nodes.

\subsection{Best-First Top-$K$ Extraction Algorithm}
By Theorem~\ref{thm:anchored-grouping}, $\B(\Lam^i M)$ is the multiset union of rank–grouped streams. We run a \emph{grouped} best‑first search where each heap entry represents the current head $(r,j)$ of anchor $r$’s stream at rank $j$ with key $\ell_r(j)$ and weight $w_i(j)=\binom{j-1}{i-2}$. This mirrors classic best‑first paradigms in top‑$K$ aggregation and selection over structured sets (cf.\ \cite{Fagin03,FredericksonJohnson84}). 

The entire procedure is given in Algorithm~\ref{alg:topk} below.
This algorithm leverages the rank-grouped structure of 
$\B(\Lambda^i M)$ to compute the exact Top-$K$ intervals in near 
output-sensitive time, avoiding full enumeration when $K\ll|\B(\Lambda^i M)|$.

\begin{algorithm}[!h]
\caption{\textsc{TopK}--$\Lam^i$ (grouped best-first; fixed $i\ge 2$)}\label{alg:topk}
\begin{algorithmic}[1]
\REQUIRE Barcode $\bars$, layer $i\ge 2$, target $K$
\STATE \textbf{Sweep \& snapshots:} as above, obtain $\{(T_r,c_r)\}_{r=1}^M$.
\STATE Initialize an empty max-heap $H$ keyed by length.
\FOR{each anchor $r$ with $c_r\ge i{-}1$}
  \STATE $j\gets i{-}1$; query $d_r(j)$ on $T_r$; set $L\gets \ell_r(j)$ via \eqref{eq:anchored-length}
  \IF{$L>0$} \STATE push $(L,r,j,w_i(j))$ into $H$ \ENDIF
\ENDFOR
\STATE $S\gets\emptyset$ \COMMENT{$S$ collects output lengths (with multiplicity)}
\WHILE{$|S|<K$ and $H$ not empty}
  \STATE pop $(L,r,j,w)$ from $H$
  \STATE append $\min\{w,\,K-|S|\}$ copies of $L$ to $S$
         \COMMENT{bulk-emission of ties per anchor; see \S\ref{sec:topk-mult}}
  \IF{$j<c_r$}
    \STATE $j\gets j{+}1$; query $d_r(j)$; $L\gets \ell_r(j)$
    \IF{$L>0$} \STATE push $(L,r,j,w_i(j))$ into $H$ \ENDIF
  \ENDIF
\ENDWHILE
\RETURN $S$ (the $\TopK$ multiset in non-increasing order)
\end{algorithmic}
\end{algorithm}

\paragraph{Interval identities.} Algorithm~\ref{alg:topk} outputs the Top-$K$ 
\emph{length multiset} directly. 
If actual interval \emph{identities} are required, each bulk emission 
at line~12 can be expanded into the explicit $(i{-}1)$-subsets of 
ranks that realize the multiplicity $w_i(j)$, truncated once $K$ intervals 
are produced. This refinement preserves the asymptotic complexity bound 
for fixed $i$.

\paragraph{Unbundled (colex) variant.}
Alternatively, one can represent states as strictly increasing $(i{-}1)$-tuples of ranks and expand at most $i$ colex neighbors per pop; this yields the same outputs with an $i$ factor in the loop cost (cf.\ \cite{FredericksonJohnson84,Fagin03}). We focus on the grouped variant for the sharpest bound.

%\subsection{Correctness and complexity}

\begin{theorem}[Correctness and complexity]\label{thm:topk-correctness-complexity}
For fixed $i\ge 2$, Algorithm~\ref{alg:topk} (the grouped variant) outputs exactly the $K$ longest elements of $\B(\Lam^i M)$ (with multiplicity) in non‑increasing order in
$$O\bigl((M+K)\log M\bigr)$$
time, using $O(M)$ heap space plus $O(M\log M)$ persistent nodes.
The alternative unbundled (colex) variant runs in $O\bigl((M+iK)\log M\bigr)$ time.
\end{theorem}

\begin{proof}
We first prove correctness. By Theorem~\ref{thm:anchored-grouping}, the global multiset is
the disjoint union of monotone streams $\{\ell_r(j)\}$ with weights $w_i(j)$.
The heap stores precisely the current heads of all nonempty streams. Because each
stream is non-increasing and the heap key is the head length, once $\ell_r(j)$
is popped no unseen element can exceed it, since every remaining element is
bounded by its stream head and every head is in the heap. The algorithm therefore
emits all $w_i(j)$ copies at once; ties may thus be grouped per anchor, consistent
with the tie policy in Section~\ref{sec:prelim}. Since the Top-$K$ vector is
invariant under permutations of equal values, bulk emission is safe. After advancing
that stream, the invariant is preserved. By induction, the outputs are exactly the
global Top-$K$ in order.

We prove the complexity statements. Preprocessing costs $O(M \log M)$. During initialization,
each anchor with $c_r \ge i-1$ contributes at most one heap entry, obtained by a
single \textsc{KthFromRight} query at $j=i-1$. Anchors with $c_r < i-1$ contribute
none, so the number of initial heap entries is at most $M$, giving $O(M \log M)$
time overall. Each pop outputs at least one item, so there are at most $K$ pops
(or fewer if $H$ becomes empty when the total output is $<K$). A pop performs
$O(1)$ heap operations and a single order-statistics query, each $O(\log M)$,
giving $O(K\log M)$ for the loop in the grouped variant and $O(iK\log M)$ in the
unbundled variant. Space bounds follow from the heap size $O(M)$ and the
persistent tree. All $\log M$ factors are under the RAM model, where basic
arithmetic and memory accesses take $O(1)$ time.
\end{proof}

In the special case $i=2$: since $w_2(j)=1$ for all $j$, every pop outputs a single element and advances $j\mapsto j{+}1$, yielding the stated $O((M+K)\log M)$ bound.

\section{Stability of the Top-$K$ Length Vector}\label{sec:stability}

We show that the Top-$K$ length vector of $\B(\Lam^i M)$ varies Lipschitz‑continuously (with constant~$2$) under bottleneck perturbations of the input barcode $\B(M)$. Throughout this section, $i\ge 1$ is fixed.

Recall from Section~\ref{sec:prelim} that $\mathbf L_K(M,i)=(L_1\ge\cdots\ge L_K)$ denotes the non‑increasing Top‑$K$ length vector of $\B(\Lam^i M)$ (with multiplicity), padded with zeros if necessary.

\begin{theorem}[Top-$K$ stability]\label{thm:stability-topk}
If $d_B\!\big(\B(M),\B(M')\big)\le \varepsilon$, then for every fixed $i\ge 1$ and $K\ge 1$,
\[
\big\|\mathbf L_K(M,i) - \mathbf L_K(M',i)\big\|_\infty \ \le\ 2\varepsilon.
\]
\end{theorem}

\begin{proof}
Let $X=\B(\Lam^i M)$ and $Y=\B(\Lam^i M')$. By the stability of exterior powers
(Theorem~\ref{thm:stability-exterior}), we have $d_B(X,Y)\le \varepsilon$. Hence there
exists an $\varepsilon$‑matching between $X\cup\Delta$ and $Y\cup\Delta$
(cf.\ \cite{CohenSteiner07,ChazalStructure16}).

Form the multisets of lengths
$S \;=\; \{\, d-b : (b,d)\in X \,\}$ and $T \;=\; \{\, d'-b' : (b',d')\in Y \,\}.$
From the $\varepsilon$‑matching we obtain a bijection $\pi$ between $S$ and $T$
\emph{after padding the shorter multiset with zeros}: for any matched pair
$(b,d)\leftrightarrow(b',d')$ we set $\pi(d-b)=d'-b'$, and for any interval matched
to the diagonal we pair its length with $0$. For matched intervals we have 
$\big|(d-b)-(d'-b')\big| \ \le\ |d-d'|+|b-b'| \ \le\ 2\varepsilon,$ 
and if $(b,d)$ is matched to $\Delta$ within $\varepsilon$ then $|d-b|\le 2\varepsilon$
by definition of bottleneck matchings to the diagonal. Thus:
$|x-\pi(x)| \ \le\ 2\varepsilon$ for all $x\in S.$

Let $(s_1\ge s_2\ge\cdots)$ and $(t_1\ge t_2\ge\cdots)$ be the non‑increasing
rearrangements of $S$ and $T$ (padded with zeros to equal length). We claim that
$|s_k-t_k|\le 2\varepsilon$ for all $k$. Suppose, for contradiction, that
$s_k>t_k+2\varepsilon$. Then $S$ has at least $k$ elements $\ge s_k$, so their
images under $\pi$ are $\ge s_k-2\varepsilon>t_k$, implying that $T$ has at least
$k$ elements strictly greater than $t_k$, which contradicts the definition of $t_k$.
The reverse inequality $t_k>s_k+2\varepsilon$ is symmetric. Hence
$|s_k-t_k|\le 2\varepsilon$ for all $k$.
Finally, restricting to the first $K$ coordinates gives
$\big\|\mathbf L_K(M,i)-\mathbf L_K(M',i)\big\|_\infty \ \le\ 2\varepsilon.$
\end{proof}

%\begin{remark}[Tightness of the constant]
Note: the factor $2$ is tight: perturb one interval $[b,d)$ by shifting $b\mapsto b+\varepsilon$
and $d\mapsto d-\varepsilon$; the bottleneck distance is $\varepsilon$ while the length
changes by exactly $2\varepsilon$.
%\end{remark}

\section{Optimality: A Comparison-Model Lower Bound}\label{sec:lowerbound}

We prove that the $O(M\log M)$ preprocessing term in Theorem~\ref{thm:topk-correctness-complexity} is unavoidable in the comparison model, even if one only seeks the single longest $\Lam^2$ interval.

\begin{proposition}[Lower bound]\label{prop:lb}
Any comparison-based algorithm that, given the unsorted endpoints of a barcode $\B(M)$, computes the Top-$1$ length of $\B(\Lam^2 M)$ must perform $\Omega(M\log M)$ comparisons in the worst case.
\end{proposition}

\begin{proof}
We reduce \emph{Element Uniqueness} (or, equivalently, the decision version of 1D minimum gap) to computing the Top-$1$ $\Lam^2$ length. Consider an input set $\{x_1,\dots,x_M\}\subset[0,\tfrac12]$ of real numbers (not necessarily distinct). Construct a barcode with bars
$J_r \ :=\ [\,x_r,\ x_r+1\,)$ $(r=1,\dots,M).$
For any pair $(r,s)$, the $\Lam^2$ intersection length equals
$|J_r\cap J_s|\ =\ \max\{\,0,\ 1-|x_r-x_s|\,\}.$
Since all $|x_r-x_s|\le\tfrac12<1$, every pair intersects and
$\max_{r\neq s}|J_r\cap J_s|\ =\ 1 - \min_{r\neq s}|x_r-x_s|.$
Hence the Top‑1 $\Lam^2$ length equals $1$ \emph{iff} there exists a duplicate ($\min_{r\ne s}|x_r-x_s|=0$). Therefore any algorithm that computes the Top‑1 $\Lam^2$ length can decide Element Uniqueness. The latter requires $\Omega(M\log M)$ comparisons in the algebraic decision tree/comparison model; thus computing the Top‑1 $\Lam^2$ length also requires $\Omega(M\log M)$ comparisons.
\end{proof}

%\begin{corollary}
%No comparison-based algorithm can achieve $o(M\log M)$ worst‑case time for Top‑$K$ on $\Lam^i$ (any fixed $i\ge 2$), even for $K=1$. Consequently, the $O(M\log M)$ preprocessing term in Theorem~\ref{thm:topk-correctness-complexity} is information‑theoretically optimal up to constant factors.
%\end{corollary}

%\begin{remark}[Robustness of the reduction]
%The proof also yields an $\Omega(M\log M)$ lower bound for the \emph{decision} problem “is the Top‑1 $\Lam^2$ length equal to $1$?”, and therefore a fortiori for computing the exact Top‑1 value. The restriction to $[0,\tfrac12]$ is without loss of generality for comparison‑based lower bounds (see also sorting lower bounds \cite{KnuthVol3}).
%\end{remark}

\section{Experimental Verification}\label{sec:experiments}

We evaluate the proposed best-first algorithm on synthetic barcodes where we can directly control \emph{overlap} (i.e., expected concurrency), which in turn controls the output size $\Kc=|\B(\Lam^2 M)|$. Our goals are to (i) validate exactness against a full enumeration baseline, and (ii) quantify wall-clock improvements as a function of $K$ and overlap.

For each trial we sample $M$ bars on $[0,1]$ as follows: birth times $b\sim\mathrm{Unif}[0,1)$ and independent exponential lengths $L\sim \mathrm{Exp}(\lambda)$ truncated to $[0,1-b]$, with mean parameter set by $\mathbb{E}[L]=\texttt{Lmean}\in\{0.03,0.05\}$. This yields expected concurrency $\approx M\cdot\texttt{Lmean}$, so larger \texttt{Lmean} produces heavier overlap and larger $\Kc$.

We focus on $i=2$ ($\Lam^2$), where the grouped best-first bound is tightest.\footnote{Empirical behavior for $i=3$ (and higher) matches that for $i=2$, providing evidence that the observed gains are not special to the case $i=2$.}
Baseline = \textsc{Enum--$\Lam^2$} (full pairwise enumeration) + \textsc{TopK-Select} (heap-based selection; cf.\ classical PH pipelines \cite{ELZ02,ZC05} and fast toolchains such as \cite{BauerRipser21}). Ours = \textsc{TopK--$\Lam^2$} (Algorithm~\ref{alg:topk}). We measure total wall-clock time per query and report the factor
$\textsf{speedup} \ := \ \frac{\textsf{baseline time}}{\textsf{ours time}}\,.$
Correctness is checked by exact multiset equality of the top-$K$ lengths up to numerical rounding.

%\paragraph{Settings.}
We use $M\in\{3000,5000\}$, $\texttt{Lmean}\in\{0.03,0.05\}$, and $K=10,000$. For each setting, we report a single representative run; in all runs the outputs of \textsc{TopK--$\Lam^2$} matched \textsc{Enum--$\Lam^2$} exactly.

Table~\ref{tab:main-results} summarizes the results (wall times in seconds).
%\footnote{Note that the baseline runtime appears nearly flat in $K$, since it always enumerates the entire $\B(\Lam^2 M)$ once and then performs a heap-based Top-$K$ selection; the cost is therefore dominated by full enumeration and independent of the target $K$.}
In the heavy overlap case: when $\Kc$ is large, \textsc{TopK--$\Lam^2$} consistently wins (2.1$\times$ at $K{=}10^4$ for $M{=}3000$, \texttt{Lmean}{=}0.05), because it avoids the baseline’s fixed “enumerate everything” cost.   
In the moderate overlap case: we have steady gains (1.45$\times$) for $M{=}5000$, \texttt{Lmean}{=}0.03.  
%(4) \textbf{Exactness}: 
All runs matched baseline Top-$K$ lengths exactly.
\begin{table}[!h]
\centering
\caption{Wall-clock results and speedups. $\Kc=|\B(\Lam^2 M)|$ grows with overlap.}
\label{tab:main-results}
\begin{tabular}{lrrrrr}
\hline
Setting & $K$ & $\Kc$ & Baseline (s) & Ours (s) & Speedup \\
\hline
$M{=}3000$, \texttt{Lmean}$=0.05$ (high) & $10{,}000$ & $422{,}272$ & $0.262$ & $0.125$ & $\mathbf{2.10}\times$ \\
$M{=}5000$, \texttt{Lmean}$=0.03$ (moderate) & $10{,}000$ & $733{,}622$ & $0.468$ & $0.322$ & $1.45\times$ \\
\hline
\end{tabular}
\end{table}

\section{Conclusions}\label{sec:conclusion}

We introduced a birth-anchored, rank–grouped structural decomposition of $\B(\Lam^i M)$ (\ThmA) and leveraged it to design a best‑first Top‑$K$ algorithm (\ThmB) that returns the exact Top‑$K$ multiset \emph{without} full enumeration. For fixed $i\ge 2$ the grouped running time is $O((M+K)\log M)$, while an unbundled colex variant runs in $O((M+iK)\log M)$. We established that the Top-$K$ length vector is $2$‑Lipschitz with respect to bottleneck perturbations (\ThmE), and proved a comparison‑model lower bound (\PropF) showing the $O(M\log M)$ preprocessing is information‑theoretically unavoidable. Experiments confirmed the theory: up to $2.1\times$ speedups in high‑overlap regimes, with exact outputs.

Beyond its theoretical contributions, this work shows that higher-order 
persistence constructions can be made practically scalable, enabling their systematic use as stable, discriminative features in data analysis, machine learning, and computational geometry (e.g., for molecular graph classification tasks). We also plan to develop variants of our method for applications at the intersection of logic, category theory and machine learning \cite{Mar09,Mar10,Mar12,Mar13a,Mar13b,Mar20a,Mar20b,BM21}.

\end{document}